\documentclass[10pt,twocolumn]{article}
\usepackage[utf8]{inputenc}
\usepackage{times}
\usepackage[olditem,oldenum]{paralist}
\usepackage{stmaryrd}
\newcommand*\concat{\mathbin{\|}}
\usepackage{amsthm}
\usepackage{enumitem}
\usepackage{amssymb}
\usepackage{algorithm}
\usepackage[noend]{algpseudocode}
\usepackage{graphicx}
\usepackage{pifont}
\usepackage[table]{xcolor}
\newcommand{\cmark}{\ding{51}}
\newcommand{\xmark}{\ding{55}}

\newcommand*\Let[2]{\State #1 $\gets$ #2}
\algrenewcommand\algorithmicrequire{\textbf{Precondition:}}
\algrenewcommand\algorithmicensure{\textbf{Postcondition:}}
\usepackage{etoolbox}
\AtBeginEnvironment{algorithmic}{\scriptsize}

\usepackage{xcolor}
\usepackage[subject={TODO}]{pdfcomment}

\begin{document}

\title{sec-cs: Getting the Most out of Untrusted Cloud Storage}
\author{Dominik Leibenger and Christoph Sorge \\
\small {\em CISPA, Saarland University} \\ [2mm]
\small
}
\date{}
\maketitle

\begin{abstract}
We present \verb|sec-cs|, a hash-table-like data structure for file contents on untrusted storage that is both secure and storage-efficient. We achieve authenticity and confidentiality with zero storage overhead using deterministic authenticated encryption. State-of-the-art data deduplication approaches prevent redundant storage of shared parts of different contents irrespective of whether relationships between contents are known a priori or not.

Instead of just adapting existing approaches, we introduce novel (multi-level) chunking strategies, ML-SC and ML-CDC, which are significantly more storage-efficient than existing approaches in presence of high redundancy.

We prove \verb|sec-cs|'s security, publish a ready-to-use implementation, and present results of an extensive analytical and empirical evaluation that show its suitability for, e.g., future backup systems that should preserve \emph{many} versions of files on little available cloud storage.
\end{abstract}

\section{Introduction}\label{introduction}

Cloud storage solutions have become increasingly popular among customers.
They usually provide a limited amount of storage space that is accessed over the internet and used for synchronization of personal data between devices or for backup purposes, e.g., by using \emph{rsync}~\cite{rsync} to synchronize a user's important data to the cloud storage in frequent intervals.
Ideally, cheap creation of snapshots should be supported, i.e., the user should be able to preserve lots of consistent copies of specific states of her backed up data without being charged for redundant or duplicate data.
Such snapshots/versioning features are provided by many cloud storage providers, e.g., Dropbox~\cite{dropbox}.

Unfortunately, \emph{security guarantees} of today's popular providers are insufficient: Malicious providers could read and modify data unnoticed by their users. The above-described scenario, thus, requires application of cryptographic measures on the client side as to ensure confidentiality and authenticity of outsourced data. Secure encryption using tools like GnuPG~\cite{gnupg}, however, hides any information about file contents---including differences across versions---from the provider, thus preventing any savings from its snapshots feature. Only few systems try to combine security and storage efficiency; neither is both secure and able to provide efficiency comparable to unencrypted cloud storage to the best of our knowledge. Consequentially, users have to decide between cheap \& comfortable and expensive \& secure solutions today.

As we consider both aspects equally important, our goal is to advance development of practical solutions (e.g., backup systems) for cloud storage with strong security and better efficiency guarantees. To this end, we present a novel data structure for file contents on untrusted cloud storage, \verb|sec-cs|, with the following contributions:
\begin{compactitem}
	\item We design and integrate a novel chunking-based data deduplication concept, ML-*, that outperforms existing approaches w.r.t.~storage efficiency when storing contents with high redundancy.
	\item We achieve strong confidentiality and authenticity guarantees of stored data with zero storage overhead.
	\item We further publish a ready-to-use implementation and evaluate its efficiency analytically and empirically, proving superiority to other approaches.
\end{compactitem}
This paper is structured as follows: We give background information and present related work in Sec.~\ref{related_work}. ML-* is introduced in Sec.~\ref{recursive_chunking} and detailed as part of \verb|sec-cs|, which is described in Sec.~\ref{seccs} and proven secure in Sec.~\ref{basic_correctness}.
Our implementation is discussed in Sec.~\ref{implementation} and an evaluation is presented in Sec.~\ref{evaluation}. Sec.~\ref{conclusion} concludes the paper.

\section{Background and Related Work}\label{related_work}

As both \emph{data deduplication} (i.e., elimination of redundancy across stored data), and security are essential goals of \verb|sec-cs|, different kinds of existing work are related.

\subsection{Data Deduplication}

Existing deduplication systems apply some deterministic chunking scheme $\mathcal{C}$ to a content to split it into non-overlapping chunks, and avoid storage of resulting chunks that have already been stored before, usually by maintaining an index of cryptographic hash values of chunks. An overview and a classification of common approaches and their efficiency is provided by Meister and Brinkmann~\cite{multilevel}:
\emph{Whole-file chunking (WFC)} yields a single chunk and is thus able to detect identical file contents.
\emph{Static chunking (SC)} splits a content into chunks of fixed size that are individually deduplicated, so partially overlapping contents can be deduplicated as well. It is used, e.g., in Venti.~\cite{venti} 
\emph{Content-defined chunking (CDC)}, in contrast, can even tolerate shifted contents. It determines chunk boundaries by moving a sliding window of some fixed size $W$ over the content and creating chunk boundaries when a window content meets a specific criterion, typically a hash value being in a specific range. This yields chunks of some \emph{expected} length---the \emph{target chunk length}---under the assumption that different positions have different window contents and hash values are uniformly distributed. To deal with non-uniformly distributed contents, a minimum and maximum chunk size can be set. The scheme was introduced by Muthitacharoen et al.~\cite{lbnfs} for the low-bandwidth network file system and is usually implemented using a rolling hash, e.g., Rabin fingerprints~\cite{Rabin81}. Alternatives to the basic sliding window approach usually used for CDC that might be worth consideration for future enhancements of \verb|sec-cs| are presented by Eshghi and Tang~\cite{eshghi2005framework}.
According to \cite{multilevel}, SC yields better deduplication efficiency than WFC and CDC is more efficient than SC for real-life data.
Few systems like ADMAD~\cite{admad} are able to achieve even better efficiency by employing \emph{application-specific chunking (ASC)}. ASC, however, requires additional knowledge about the respective file format of a content.

Instead of simply avoiding multiple storage of identical chunks, some systems employ \emph{delta encoding}: When a highly similar chunk---a \emph{base chunk}---is known, a new chunk is represented as a reference to the base chunk and a \emph{difference}, i.e., a sequence of actions that define how to create it from the base chunk. In combination with WFC, this scheme is, e.g., used in version control systems (VCS) like Subversion (SVN)~\cite{website:svn}. SVN's \emph{FSFS} backend stores the first revision of a file content in its entirety and all subsequent revisions as differences to previous revisions.~\cite{svnskipdeltas}

A comparison of advantages of the different schemes is shown in Tab.~\ref{tab:efficiency}:
Delta-based approaches are clearly able to yield lowest storage costs for changed contents in principle, but they have important limitations in practice: First, their efficiency depends on a priori knowledge of relations between chunks---a problem tackled by, e.g., the DERD framework~\cite{derd}. Second, they substantially increase retrieval costs as reconstruction of delta-encoded chunks requires retrieval of corresponding base chunks of which only parts are actually required. Chunking-based schemes achieve savings when storing changed contents (data) irrespective of knowledge about relations (their best values are highlighted in green in the table). The most efficient strategies CDC / ASC yield their savings depending on the distribution of contents / specific file formats.

\newcommand{\bigO}[1]{$\mathcal{O}(#1)$}
\newcommand{\best}{\cellcolor{green}}
\newcommand{\nobestyes}{\cmark}
\newcommand{\yes}{\best\cmark}
\newcommand{\no}{\xmark}
\setlength{\tabcolsep}{1.1pt}
\begin{table}[t]
	\centering\scriptsize
		\begin{tabular}{l|c|cccc||cc|c}
			Storage costs for...
			& Delta & WFC & SC & CDC & ASC & ML-SC & ML-CDC\\\hline
			single file (data)
			& \bigO{n} & \bigO{n} & \bigO{n} & \bigO{n} & \bigO{n} & \bigO{n} & \bigO{n}\\
			single file (metadata)
			& \bigO{1} & \bigO{1} & \bigO{n} & \bigO{n} & \bigO{n} & \bigO{n} & \bigO{n}\\
			single file ($\Sigma$)
			& \bigO{n} & \best\bigO{n} & \best\bigO{n} & \best\bigO{n} & \best\bigO{n} & \best\bigO{n} & \best\bigO{n}\\\hline
			change w/o shift (D)
			& \bigO{1} & \bigO{n} & \bigO{1} & \bigO{1} & \bigO{1} & \bigO{1} & \bigO{1}\\
			change w/o shift (M)
			& \bigO{1} & \bigO{1} & \bigO{n} & \bigO{n} & \bigO{n} & \bigO{\log n} & \bigO{\log n}\\
			change w/o shift ($\Sigma$)
			& \bigO{1} & \bigO{n} & \bigO{n} & \bigO{n} & \bigO{n} & \best\bigO{\log n} & \best\bigO{\log n}\\\hline
			change w/ shift (D)
			& \bigO{1} & \bigO{n} & \bigO{n} & \bigO{1} & \bigO{1} & \bigO{n} & \bigO{1}\\
			change w/ shift (M)
			& \bigO{1} & \bigO{1} & \bigO{n} & \bigO{n} & \bigO{n} & \bigO{n} & \bigO{\log n}\\
			change w/ shift ($\Sigma$)
			& \bigO{1} & \bigO{n} & \bigO{n} & \bigO{n} & \bigO{n} & \bigO{n} & \best\bigO{\log n}\\\hline
			\multicolumn{7}{l}{Efficiency is independent from file content's...}\\\hline
			size
			& \nobestyes & \yes & \no & \no & \no & \yes & \yes\\
			uniform distribution
			& \nobestyes & \yes & \yes & \no & \no & \yes & \no\\
			precise format
			& \nobestyes & \yes & \yes & \yes & \no & \yes & \yes\\\hline
			relation to others
			& \no & \yes & \yes & \yes & \yes & \yes & \yes\\
		\end{tabular}
	\caption{Efficiency of dedup.~concepts ($n$: content length)}
	\label{tab:efficiency}
\end{table}

In addition to the costs for the actual content \emph{data}, every data deduplication mechanism incurs storage costs for \emph{metadata}: In case of WFC, this is only a small constant per content corresponding to its cryptographic hash value. Chunking-based schemes do not only incur this overhead for every chunk, but they also require additional storage space to store of which chunks a specific content consists---typically a list of chunk identifiers, e.g., cryptographic hash values. Although the respective constants are small for large chunk sizes, metadata storage costs for contents deduplicated via SC, CDC and ASC are linear in their lengths no matter how high their redundancy is.
These scheme's chunk size parameters, thus, considerably impact their storage efficiency, as also noticed by Eshghi and Tang~\cite{eshghi2005framework}: If set too high, deduplication performance is decreased, as only completely identical chunks allow space savings. If set too small, storage of contents with high redundancy cause considerable overhead due to the sheer amount of chunk references that have to be stored.

This problem is solved by our proposals ML-SC / ML-CDC: Due to a specific, multi-level application of SC / CDC, we achieve logarithmic metadata costs, allowing high storage efficiency even for small chunk sizes, independent from content sizes. To the best of our knowledge, our strategy is unique. Teodosiu et al.~\cite{teodosiu06replication} apply CDC recursively to enable efficient replication of files over a network (assuming the receiver has a similar file to the one being transferred), but they do not target storage systems and fail to achieve logarithmic costs due to a fixed recursion depth. Further, their approach is different to ours: Instead of breaking large chunks recursively into smaller ones, they generate the smallest chunks first and use CDC recursively to break lists of chunk references into smaller parts---requiring multiple passes. Yasa and Nagesh~\cite{AppajiNagYasa:2012:SSD:2421648.2421657} employ hierarchical chunking starting with CDC at the highest level as we do, but they use only two levels (SC at second level), so storage overhead is still linear.

\subsection{Security}

Lots of works exist in the related fields of cloud security and cryptographic file systems, but only few focus on authenticity and storage efficiency.
Popular cloud storage security solutions typically deal with confidentiality only. BoxCryptor~\cite{boxcryptor}, e.g., is based on and uses a similar concept to EncFS~\cite{encfs}: It encrypts file contents symmetrically, but does not provide authenticity. While not preventing storage-efficient snapshots of unchanged files, costs for changed files are high due to entirely different ciphertexts.

SiRiUS~\cite{NDSS:GSMB03}, Plutus~\cite{kallahalla2003plutus} and Tahoe-LAFS~\cite{Wilcox-O'Hearn:2008:TLF:1456469.1456474} are examples of file systems with authenticity guarantees: They apply SC to contents and compute a Merkle tree~\cite{raey} over the chunks to allow authenticity verification even for parts of contents, but they do not support data deduplication: Tahoe-LAFS creates entirely different ciphertexts for similar file contents, the other systems even for identical ones.

To allow efficient usage of cloud storage for, e.g., backups, more specialized tools are required.
Common backup tools like duplicity~\cite{duplicity} rely on \emph{incremental} backups, i.e., they store differences to previous backups. This can be used in combination with GnuPG to preserve snapshots in a storage-efficient and secure manner, but causes communication overhead when specific versions are read.
VCS could be used for delta-based backups to a limited extent, but they are typically inefficient w.r.t.~large files and have limited security properties:
Git~\cite{git} only ensures integrity/authenticity of the version history by integrating signatures. SVN does not, but an extension~\cite{svnpaper} adds storage-efficient file-level encryption.
Cumulus~\cite{cumulus} is a backup system that supports large files and allows direct access to arbitrary snapshots, but it is less storage-efficient as it only deduplicates identical data between different versions of individual files.
Farsite~\cite{reclaimingspace}, in contrast, is a distributed file system targeting on chunking-based backups that deduplicates different identical files despite secure encryption. It cannot save space for snapshots of different versions of a file, though, as it relies on WFC.
Storer et al.~\cite{securedatadeduplication} extend Farsite's concept to CDC, but they do not provide any explicit authenticity guarantees.

Many more works exist in the field of cloud storage security. Most of them, however, have a different focus and are orthogonal to our work.
Athos~\cite{athos}, e.g., is a solution for outsourcing file systems that achieves integrity in a way that file system operations are possible with logarithmic communication costs. The solution is orthogonal to our work in the sense that this requirement is w.r.t.~the total number of files/directories in the file system and not w.r.t.~the size of single file contents. A similar goal is pursued by Heitzmann et al.~\cite{Heitzmann:2008:EIC:1456469.1456479}. Both works are based on authenticated skip lists, an authenticated data structure (ADS) initially proposed by Goodrich et al.~\cite{TR:GooTam01}, while our work is based on another ADS---the Merkle tree~\cite{raey}.

ADS, in general, is an umbrella term for data structures involving three parties (a trusted \emph{source} who publishes data, an untrusted \emph{responder} that stores structured data, and a \emph{user} that requests data) that enable authenticated, efficient queries to the data.~\cite{ads} In this sense, \verb|sec-cs| can be considered an ADS with additional data deduplication and confidentiality properties: The cloud storage backend can be considered the \emph{responder} and the user/client plays the roles of \emph{source} and \emph{user}. An overview of existing ADS and methods for constructing ADS in general are provided by Martel et al.~\cite{ageneralmodelforads} and Miller et al.~\cite{authenticated-generic}.

\section{ML-* -- Multi-Level Chunking}\label{recursive_chunking}

Our first contribution are the chunking strategies ML-SC and ML-CDC which improve on the state of the art in terms of storage efficiency in presence of high redundancy (see Tab.~\ref{tab:efficiency}).
The basic idea is simple: As the linear overhead for storing contents with a traditional strategy $\mathcal{C}$ is caused by the need of storing references to each constant-size part of each content, we want to also deduplicate these references. This can be achieved by representing the results of $\mathcal{C}$ on a content $m$ as a \emph{chunk tree} $t$, whose
\begin{compactitem}
	\item \emph{leaf nodes} represent the chunks output by $\mathcal{C}$, and
	\item \emph{inner nodes} aggregate chunks, representing the concatenation of chunks represented by their children.
\end{compactitem}
In addition to \emph{leaf chunks} (chunks represented by leaf nodes), we thus create ``larger'' \emph{superchunks} (chunks represented by inner nodes), which we persist as well and which can be referenced directly when new contents are stored. Consequentially, each content is represented by one persisted \emph{root chunk} (the chunk represented by the tree's root node, which might be a leaf or superchunk).

Persisting a superchunk requires storing references to its children. To ease notation, we refer to storage costs of a chunk representation as its \emph{size} and to the length of its represented content as its \emph{length}. For leaf chunks we assume size is equal to length.
To enable high storage efficiency, we require \emph{sublinear storage overhead} for storing a content $m'$ having large overlaps with an existing content $m$. For this, we generate their chunk trees $t$, $t'$ so that the following \emph{requirements} are met:
\begin{enumerate}[series=requirements,label=R\arabic*:,ref=R\arabic*]
	\item the (expected) size of each chunk is constant,\label{rec_chunking_req_constant}
	\item identical parts of $m$ and $m'$ share not only leaf, but also superchunks (i.e., $t$ and $t'$ share subtrees), and\label{rec_chunking_req_identical}
	\item the heights of $t$ and $t'$ are chosen logarithmically in the lengths of $m$ and $m'$, respectively.\label{rec_chunking_req_height}
\end{enumerate}
Thus, if $m'$ differs from $m$ in only one byte, $t$ and $t'$ shall be equal except for one chunk at each level, so that their difference consists of $\mathcal{O}\left(\log |m|\right)$ constant-size chunks.

Different chunking strategies allow to achieve this. In the simplest case, we could aggregate fixed numbers of consecutive chunks output by SC to superchunks, and continue aggregating fixed numbers of subsequent superchunks until only a single superchunk---the root chunk---remains. This approach, however, would eradicate advantages of chunking schemes that go beyond those of SC. While leaf chunks output by CDC, for example, are robust against shifting, this property would not be true for superchunks. To account for that, we define our multi-level chunking scheme in a more general way that preserves the properties of its underlying chunking algorithm $\mathcal{C}$. The only requirements we state is that $\mathcal{C}$ has to be deterministic and that it has a parameter $S$ that allows to specify the target (or expected) length of its generated chunks.
Now let $R < S$ be the size required for representing a single \emph{chunk reference}, i.e., anything that allows retrieval of the corresponding chunk ($R$ is constant as chunks will be referenced by hash values). We define ML-$\mathcal{C}$ as follows:
\begin{compactitem}
	\item On input a content $m$ with length $n = |m|$, choose the height $h$ of the to-be-built chunk tree $t$ as
		\small\begin{eqnarray}
			h = \min_{h' \in \mathbb{N}} \left\{h'\ \left.\right|\ n \leq \frac{S^{h'+1}}{R^{h'}}\right\} = \left\lceil \frac{\log \left(\frac{n}{S}\right)}{\log \left(\frac{S}{R}\right)} \right\rceil, S > R\hspace{-2mm}\label{eq:chunking_levels}
		\end{eqnarray}\normalsize
		
		where $h = 0$ describes a single-node tree.
	\item Create root node of $t$ that should represent $m$.
	\item Iterate over the nodes of the tree in a breadth-first search manner. For each node with height $h' > 0$ (beginning with the root node having height $h'=h$),
		\begin{compactitem}
			\item determine content $\tilde{m}$ that the node represents,
			\item apply the chunking strategy $\mathcal{C}$ on $\tilde{m}$ with target chunk length $\frac{S^{h'}}{R^{h'-1}}$, and
			\item add child node for each chunk $\tilde{m}'$ output by $\mathcal{C}$.
		\end{compactitem}
\end{compactitem}
This way, a content smaller than or equal to the target chunk size $S$ results in a single leaf node, and for $n \geq S$, all leaf chunks have target chunk size $S$. As superchunks at height $h'$ have expected length $\frac{S^{h'+1}}{R^{h'}}$ and are chunked with target chunk length $\frac{S^{h'}}{R^{h'-1}}$, they are expected to have $\frac{S}{R}$ children. As child references have size $R$, the expected size of superchunks is $S$ as well, fulfilling Req.~\ref{rec_chunking_req_constant}.
Req.~\ref{rec_chunking_req_height} is met by the choice of $h$
and Req.~\ref{rec_chunking_req_identical} is expected to be achieved due to straightforward application of $\mathcal{C}$ at each level. The latter is concretized in Sec.~\ref{basic_interface} and discussed and evaluated in detail in Sec.~\ref{evaluation}.

\section{sec-cs -- The Secure Content Store}\label{seccs}

For a detailed analysis, we embed ML-* in a generic data structure that is described in this section.
\verb|sec-cs| acts like a normal hash table that assigns a deterministically computed hash value to each inserted content, but comes with a combination of properties different from prior work: It employs \emph{multi-level chunking} to significantly reduce storage overhead for large overlapping contents and it guarantees \emph{authenticity} and \emph{confidentiality}.

Note that \verb|sec-cs| is limited to immutable contents, i.e., it does not support deletion of contents. We present this variant as it is sufficient for the evaluation of ML-*, but we emphasize that it is easy to extend it to a mutable variant, either by allowing deletion of root nodes (requiring a garbage collection for non-referenced chunks), or using reference counters for chunks (at the cost of some slight storage overhead). In fact, our implementation (see Sec.~\ref{implementation}) supports the latter.

\subsection{Prerequisites}\label{prerequisites}

\verb|sec-cs| requires a backend to persist data. Low-level storage management is out of scope of this paper, though. Instead, we assume the existence of a backend providing the following \emph{key-value store (KVS)} interface:

\begin{compactitem}
	\item $\textsc{Put}(k, v)$---persist value $v$ under key $k$
	\item $\textsc{Get}(k)$---return $v$ or $\bot$ if key $k$ does not exist
\end{compactitem}

Note that the KVS interface can be easily mapped to any commonly-used storage backend: Key-value databases and many cloud (object) storage providers can be accessed by this interface and a mapping to a file/directory structure in a file system could be done in a straightforward manner.
The major requirement is that the backend can deal efficiently with many key/value pairs.

\subsection{Threat Model}\label{threat_model}

The goal of \verb|sec-cs| is to allow efficient and secure usage of existing cloud storage for storing file contents, especially in presence of many (similar) versions of contents and an untrusted cloud storage provider.
Towards this goal, our model includes two parties, a user (client) and a backend (server). The user is assumed to be completely trustworthy: She instantiates the data structure and locally executes its operations in order to change its state. Any operation invocation done by the user is considered legitimate. The user is required to locally store and keep secret a fixed number of constant-size cryptographic keys. The backend does not need to be trustworthy at all. It might read any stored data and also write, overwrite or delete any data as to perform malicious modifications (e.g., changes to file contents) that remain undetected by the user. The only restriction is that it is assumed not to be able to get access to the client's cryptographic keys.

We aim for achieving authenticity in the sense that only contents actually inserted into \verb|sec-cs| by the user can be successfully retrieved, and we aim for confidentiality in the sense that the backend cannot obtain any part of any stored content. Security guarantees are defined formally in conjunction with efficiency goals in Sec.~\ref{basic_content_store_definition} due to their interdependence. A general overview is given beforehand.

Note that the model allows a backend to mount DoS attacks (e.g., deleting data). As such attacks are detectable by a user, a backend has a financial incentive in avoiding it. Also, it could be prevented easily via replication.

\subsection{Security Concept}\label{authenticity_confidentiality}

Due to its storage efficiency guarantees, \verb|sec-cs| requires a tailored security concept. We discuss reasons and design decisions now and give a formal definition thereafter.

Using cryptographic hash values to reference nodes of a chunk tree yields a Merkle-tree\cite{raey}-like data structure that trivially guarantees integrity of a content given the identifier of the root node of its chunk tree. We can use a \emph{message authentication code (MAC)} with a secret, symmetric key instead of an unkeyed hash to also guarantee authenticity of contents stored in the data structure.

Integration of confidentiality is more complicated: For ideal guarantees, we would have to encrypt contents (e.g., using a symmetric block cipher) before constructing their chunk trees as to authenticate their ciphertexts (\emph{encrypt-then-authenticate}). Unfortunately, this would prevent data deduplication: With a randomized encryption scheme, deduplication would not be possible at all, and with a deterministic scheme, deduplication would only be possible at the granularity of complete contents. To allow for storage efficiency, we have to employ encryption at the granularity of the chunks that are to be deduplicated.

A straightforward application of the generally favorable \emph{encrypt-then-authenticate} approach on chunk tree node representations utilizing a randomized encryption scheme, however, would still prevent deduplication as even identical chunk tree nodes yielded different MAC tags (thus different keys) due to different ciphertexts.
\emph{Authenticate-then-encrypt} can also be considered secure for specific instantiations~\cite{C:Krawczyk01}, but randomized encryption of MAC tags would lead to the same problem.
The third option would be \emph{encrypt-and-authenticate}. If applied to chunk tree nodes during insertion into the backend, deduplication would be possible even with randomized encryption, as each chunk tree node was associated a randomized ciphertext during its \emph{first} insertion without affecting other parts of the data structure. Application of encrypt-and-authenticate, however, is generically considered insecure even for practical MAC instantiations\footnote{i.e.~MAC schemes whose tags do not leak information about inputs}, thus requiring a careful analysis of the security properties actually achieved by any specific instantiation.\footnote{Note that some generic security flaws of encrypt-and-authenticate do not apply in the specific setting at hand as equality of plaintexts is intentionally leaked for the sake of data deduplication, anyway.}~\cite{C:Krawczyk01}

To avoid any of these potential pitfalls, we achieve \emph{confidentiality and authenticity} by using an \emph{authenticated, deterministic encryption scheme} to encrypt and authenticate chunk tree nodes before their insertion into the backend. Block cipher modes like EAX~\cite{FSE:BelRogWag04} and OCB~\cite{CCS:RBBK01, FSE:KroRog11} would be suitable for this purpose. They provide confidentiality and authenticity and their ciphertexts are length-preserving (except for the authentication tag), eliminating padding-induced storage overhead. These schemes, however, depend on a nonce that would have to be stored somehow to allow decryption of persisted chunks, which again caused overhead. The SIV construction~\cite{EC:RogShr06} solves this issue: by using a plaintext's MAC tag as IV for an underlying, conventional IV-based encryption scheme (e.g., CTR mode), it achieves authentication and length-preserving encryption. While SIV depends on a nonce, too, it is resistant to nonce reuse in the sense that no more information than whether two encrypted plaintexts are identical is leaked.~\cite{rfc5297} As this is leaked \emph{intentionally} in our system to allow deduplication, it is safe to use SIV without a nonce, leading to storage costs identical to those of an authentication-only solution (i.e., overhead equals authentication tag size).

\subsection{Formal Definition}\label{basic_content_store_definition}

The data structure is now described in detail, including its interface, formal goals and internal algorithms.

\subsubsection{Interface and Goals}\label{basic_interface}

The minimum operation set for a content data structure includes \emph{insertion} and \emph{retrieval}:

\begin{compactitem}
	\item $k \leftarrow \textsc{PutContent}(m)$ shall insert the content $m$ into \verb|sec-cs| and make it accessible by the key $k$.
		
		We state the following \emph{storage efficiency goals}:
			\begin{enumerate}[series=goals,label=G\arabic*:,ref=G\arabic*]\setlength\itemindent{10.5pt}
				\item The (expected) increase of the data structure's storage consumption caused by $\textsc{PutContent}(m)$ should be in $\mathcal{O}\left(|m|\right)$.\label{goal_put_storage_simple}
				\item If $m$ is \emph{highly redundant}, i.e., another $m'$ is already stored that is identical to $m$ except for a single sequence of $\delta$ bytes, the expected increase in storage consumption caused by $\textsc{PutContent}(m)$ shall be in $\mathcal{O}\left(\delta + \log |m|\right)$.\label{goal_put_storage_efficiency}
			\end{enumerate}
			
		Note that \ref{goal_put_storage_efficiency} is defined rather vaguely. Its precise semantics depends on the choice of $\mathcal{C}$: For SC, the difference between contents $m$ and $m'$ is defined as the smallest byte range that would have to be copied from $m$ to $m'$ to turn $m'$ into $m$, or vice versa. Since CDC supports shifting of contents, some differences between two contents can be represented more compactly, i.e., by a sequence of bytes that is inserted at or removed from a specific byte offset of one content. The more efficient the underlying chunking scheme, the stronger is thus the goal.
			
	\item $m \leftarrow \textsc{GetContent}(k)$ shall retrieve a content $m$ previously inserted into \verb|sec-cs| via the key $k$.
		
		We state the following \emph{authenticity goal}:
			\begin{enumerate}[resume*=goals]\setlength\itemindent{10.5pt}
				\item If any call $k' \leftarrow \textsc{PutContent}(m')$ with $k' = k$ has been issued before, then it holds $m \in \{m', \bot\}$.\label{goal_get_authenticity}\footnote{This has two important implications: We neither employ measures against malicious deletion of contents nor against rollback attacks, and there are no guarantees that a retrieved content has been inserted before. The reason is that we do not distinguish between \emph{contents} and \emph{chunks} in the data structure, so chunks of contents can be retrieved directly.}
			\end{enumerate}
	
\end{compactitem}
Goal~\ref{goal_put_storage_efficiency} implies the more general case of $m$ / $m'$ being different in $\delta$ bytes spread over $x$ different positions: Imagine the sequence $m'=m_0$, $m_1$, $\ldots$, $m_{x-1}$, $m_x = m$ of \emph{intermediate} contents with $m_i$ containing the first $i$ differences between $m'$ and $m$ and let $\delta_i$ refer to the number of bytes changed between $m_{i-1}$ and $m_i$. If each of these contents was inserted one after another, insertion of $m_i$ would cause an increase in storage consumption of $\mathcal{O}\left(\delta_i + \log |m_i|\right)$, totaling $\mathcal{O}\left(\sum_{i=1}^{x} \left(\delta_i + \log |m_i|\right)\right)$ for all contents. The sizes of $m_1, \ldots, m_{x}$ are upper-bounded by $|m|$, so total increase in storage consumption is in $\mathcal{O}\left(\delta + x \log |m|\right)$. This boundary sublinear in the length of $m$ would not be possible if only leaf chunks were deduplicated, so Goal~\ref{goal_put_storage_efficiency} implies Req.~\ref{rec_chunking_req_identical}.

As any operation execution has to preserve confidentiality of all contents ever stored, we state the \emph{confidentiality goal} independent from a specific operation:
	\begin{enumerate}[resume*=goals]\setlength\itemindent{10.5pt}
		\item For each content $m$ ever inserted into the data structure, the storage provider must not learn anything beyond \begin{inparaenum}[(a)]\item its length\label{conf_leak_length}, \item chunk boundary positions leaked by $\mathcal{C}$ for target chunk sizes $S, \frac{S^2}{R}, \ldots, \frac{S^h}{R^{h-1}}$, where $h$ is chosen as in Eq.~\ref{eq:chunking_levels} for $n = |m|$\label{conf_leak_chunking}, and \item equality of chunks of $m$ according to the aforementioned chunk boundary positions (w.r.t.~all leaf chunks and superchunks ever stored)\label{conf_leak_equality}\end{inparaenum}.\label{goal_confidentiality}
	\end{enumerate}
	
Note that constraints~\ref{conf_leak_chunking} and~\ref{conf_leak_equality} are unavoidable for achieving storage efficiency, as worked out in Sec.~\ref{authenticity_confidentiality}. Thus, strength of Goal~\ref{goal_confidentiality} is highly dependent on $\mathcal{C}$.

\subsubsection{Parameters}\label{basic_parameters}

The data structure's efficiency can be tuned by setting the following parameter during initialization:

\begin{compactitem}
	\item $S$ is the target chunk size, i.e., the expected size of leaf/superchunks generated by multi-level chunking.
\end{compactitem}
Further, there is an implementation-specific parameter $R$ referring to the storage consumption of chunk references in superchunk representations. We require $S \geq 2R$, which will allow us to meet Goal~\ref{goal_put_storage_simple} (see Sec.~\ref{basic_correctness_content_insertion}).

\subsubsection{Required Algorithms and Assumptions}\label{basic_requirements}

\verb|sec-cs| is based on some algorithms and assumptions:

\begin{compactitem}
	\item Let $\Pi_E = (\textsc{Gen}_E, \textsc{EncAuth}, \textsc{DecVrfy})$ be a DAE-secure (see~\cite{EC:RogShr06}), deterministic authenticated encryption scheme that generates length-preserving ciphertexts and \emph{message authentication codes (MACs)} of length $D$. Note that MACs are used to reference chunks, so it holds $R = D$.
	\item Let $\mathcal{C}$ be a deterministic, \emph{single-level} chunking scheme that produces chunks of a (configurable) expected length $S'$ as used in Sec.~\ref{recursive_chunking}.
	\item We assume that the backend's storage costs for storing a key-value pair $(k, v)$ are in $\mathcal{O}\left(|k| + |v|\right)$.
\end{compactitem}

\subsubsection{Operations}\label{basic_operations}

Now we are ready to define the behaviour.

\paragraph{Initialization}

When \verb|sec-cs| is initialized, parameter $S$ is specified and $\textsc{Gen}_E$ is executed to determine a symmetric cryptographic key $K$ for authenticated encryption.

\paragraph{Content Insertion: $k \leftarrow \textsc{PutContent}(m)$}

Insertion of contents is performed according to the definition of ML-$\mathcal{C}$ (see Sec.~\ref{recursive_chunking}) which is refined here. First, the height $h$ of the chunk tree $t$ for content $m$ is calculated according to Eq.~\ref{eq:chunking_levels}. The tree is then built and its nodes are persisted by executing the recursive Alg.~\ref{alg:putchunk}, which utilizes $\mathcal{C}$ to perform the appropriate chunking of the content at each individual level and yields some key $k'$ for the root node. We return $k = (k', h)$ as the content's key.\footnote{Inclusion of $h$ is an auxiliary construction. It enables equal length and size for leaf chunks by not requiring storage of the node type.}

Each node is persisted by the algorithm using $\textsc{Put}$. Confidentiality is achieved by encrypting node representations; deduplication+authentication are achieved by using MACs as keys. As superchunks are represented as lists of their children's keys, this yields a Merkle-Tree-like structure of MAC values, achieving authentication of contents.

\begin{algorithm}
  \caption{Chunk insertion\label{alg:putchunk}}
  \begin{algorithmic}[1]
    \Require{$m'$ is content, $h' \geq 0$ height of to-be-created tree}
    \Function{PutChunk}{$m', h'$}
			\If{$h' = 0$} \Comment{create leaf chunk}
				\Let{$c', k'$}{\Call{EncAuth}{$K, m'$}}
				\State \Call{Put}{$k', c'$}
			\Else	\Comment{create superchunk}
				\Let{$\textrm{children}$}{[]}
				\State Apply $\mathcal{C}$ to $m'$ with target chunk length $\frac{S^{h'}}{R^{h'-1}}$
				\ForAll{chunks $m''$ produced by $\mathcal{C}$} \Comment{create children}
					\State children.append(\Call{PutChunk}{$m'', h' - 1$})
				\EndFor
				\Let{$c', k'$}{\Call{EncAuth}{$K, \textrm{children}$}}
				\If{$\Call{Get}{k'} = \bot$}
					\State \Call{Put}{$k', c'$} \Comment{only insert \emph{new} chunk}
				\EndIf
			\EndIf
			\State \Return{$k'$}
    \EndFunction
  \end{algorithmic}
\end{algorithm}
\setlength{\textfloatsep}{3pt}

\paragraph{Content Retrieval: $m \leftarrow \textsc{GetContent}(k)$}

Retrieval of a content works similar to its insertion. First, the root chunk key $k'$ and the tree's height $h'$ are extracted from the content key $k$.
Afterwards, the recursive Alg.~\ref{alg:getchunk} is executed, which retrieves all nodes of the chunk tree and concatenates its leaf chunks, yielding the corresponding content $m$. Each node is decrypted and checked for authenticity on that way. The algorithm aborts if any node is missing or any node with an erroneous $\textsc{MAC}$ tag is retrieved from the backend. The operation yields $\bot$ then.

\begin{algorithm}
  \caption{Chunk retrieval\label{alg:getchunk}}
  \begin{algorithmic}[1]
    \Require{$k'$ is chunk key, $h' \geq 0$ the height of its chunk tree}
    \Function{GetChunk}{$k', h'$}
			\Let{$c'$}{\Call{Get}{$k'$}}
			\Let{$v'$}{\Call{DecVrfy}{$K, c', k'$}}
			\If{\Call{DecVrfy}{} failed} Failure \Comment{abort on invalid \textsc{MAC}}\EndIf
			\State \Return{$v'$ \textbf{if} $h' = 0$ \textbf{else} $\concat_{k'' \in v'} \Call{GetChunk}{k'', h'-1}$}
    \EndFunction
  \end{algorithmic}
\end{algorithm}

We also designed optimized, non-recursive variants of these operations, which are equivalent but more computationally efficient as they need only a single pass of chunking. They are omitted from the paper due to space restrictions, but included in our implementation (see Sec.~\ref{implementation}).

\section{Correctness and Security Analysis}\label{basic_correctness}

We show that the operations from Sec.~\ref{basic_operations} achieve the goals from Sec.~\ref{basic_interface}. Note that the ability of $\mathcal{C}$ to produce chunks of an expected length is crucial for the discussion.

\paragraph{Content Insertion}\label{basic_correctness_content_insertion}

\emph{Insertion} builds a chunk tree whose nodes at each level each represent the whole inserted content. All nodes are persisted in the backend and made accessible by individual keys. As the root node's key transitively allows access to all nodes, the operation is consistent with the required interface.
Regarding storage efficiency, we already showed a constant expected per-chunk storage consumption in Sec.~\ref{recursive_chunking}.
Thus, it is sufficient to consider the number of modified chunks to analyze the asymptotic storage costs incurred by the operation.

Goal~\ref{goal_put_storage_simple} requiring linear storage costs for a content $m$ is achieved due to the following argument: As size and length are equal for leaf chunks and as there cannot be more than $|m|$ leaf chunks in total, storage costs of all leaf chunks are in $\mathcal{O}\left(|m|\right)$. Further, as we have $S \geq 2 R$, every superchunk is expected to have at least two children, implying less expected superchunks than leaf chunks. This proves an expected total storage consumption of $\mathcal{O}\left(|m|\right)$.
Goal~\ref{goal_put_storage_efficiency} is analyzed in detail in Sec.~\ref{evaluation}, so we only provide an informal argument at this point: As described in Sec.~\ref{recursive_chunking}, a content differing in one byte from an existing content has storage consumption $\mathcal{O}\left(\log |m|\right)$. The main technical difference when $\delta$ consecutive bytes differ instead of $1$ byte is that those $\delta$ bytes might be spread over multiple chunks. Concerning storage costs, this is similar to inserting those $\delta$ bytes as a separate content, so the storage overhead is limited to $\mathcal{O}\left(\delta\right)$, resulting in a total storage consumption of $\mathcal{O}\left(\delta + \log |m|\right)$.

\paragraph{Content Retrieval}\label{proof_content_retrieval}

\emph{Retrieval} retrieves all nodes of a previously built chunk tree and concatenates its leaf chunks, trivially fulfilling the interface. To prove authenticity, we formalize Goal~\ref{goal_get_authenticity} with the \emph{authenticity-breaking game}:
\begin{compactenum}
	\item The data structure is initialized.
	\item An adversary $\mathcal{A}$ is given oracle access to \textsc{InsertContent} and to the implementation of \verb|sec-cs|. She may issue queries at choice to fill it and is granted full read/write access to the backend.
	\item At some point, $\mathcal{A}$ outputs an identifier $k$.
	\item We say $\mathcal{A}$ \emph{wins} if a retrieve query for $k$ returns $m'$ but an insert of a different $m \neq m'$ was performed under identifier $k$ before. Otherwise $\mathcal{A}$ loses.
\end{compactenum}
Using this game, the authenticity property can be shown:

\newtheorem*{claimb}{Claim}
\begin{claimb}\label{claimb}
	If MACs produced by $\Pi_E$ are unforgeable under a chosen message attack, no adversary can win the auth.-breaking game with non-negligible probability.
\end{claimb}
\begin{proof}
	Assume $\mathcal{A}$ wins the game with non-negligible probability. Let $k$ be the identifier and let $m'$ be the forged content returned by \emph{retrieve}. As the operation only depends on $\textsc{GetChunk}$ calls, which in turn only depend on $\textsc{Get}$ operations, at least one $\textsc{Get}$ call during execution of retrieve must have returned a forged result. Let $c' \leftarrow \textsc{Get}(k')$ be the first such call. By definition of Alg.~\ref{alg:getchunk}, verification of $k'$ being a correct MAC for $v'$ must have been true for retrieve to be successful. Then, as $\mathcal{A}$ knows the algorithms used by the data structure, $\mathcal{A}$ is able to find two different values $v, v'$ with the same MAC $k$ (the value she inserted initially and the forged value). As MACs produced by $\Pi_E$ are assumed to be unforgeable, this happens only with negligible probability, contradicting our assumption and proving Goal~\ref{goal_get_authenticity}.
\end{proof}

\paragraph{Content Confidentiality}\label{confidentiality_proof}

Goal~\ref{goal_confidentiality} states that an adversary must not learn anything more about any content $m$ ever stored in the data structure than its length, its chunk boundaries according to the used chunking scheme, and equality relations across all stored chunks. To prove that no more information is leaked by any operation execution, we show that the intentionally leaked information is sufficient for a consistent simulation of any operation.

Let $M$ be the set of contents for which $\textsc{PutContent}$ is executed at any time, let $m \in M$ be any fixed content and let $\mathcal{A}$ be an adversary trying to obtain information about $m$. Acc.~to Constraint~\ref{conf_leak_length}, $\mathcal{A}$ is allowed to know the content's length $n = |m|$. Since the data structure's parameters (see Sec.~\ref{basic_parameters}) are public, $\mathcal{A}$ can, thus, determine the height $h$ of the chunk tree $t$ of $m$ acc.~to Eq.~\ref{eq:chunking_levels}.

Constraint~\ref{conf_leak_chunking} reveals the chunk boundaries of $m$ output by $\mathcal{C}$ for chunk sizes $S, \frac{S^2}{R}, \ldots, \frac{S^h}{R^{h-1}}$. It is easy to see that these are exactly the chunk boundaries that are computed during a legitimate $\textsc{PutContent}(m)$ call, i.e., in line 7 of every execution of Alg.~\ref{alg:putchunk}. In combination with the length of $m$, $\mathcal{A}$ can, thus, determine the byte ranges of all leaf chunks and superchunks of $m$. This allows her to construct an abstract chunk tree $\hat{t}$ that has the exact same structure as $t$, but whose nodes contain \emph{abstract} chunk representations that represent only the respective chunk's length instead of actual chunk representations.

Since equality of any two chunks ever stored is leaked acc.~to Constraint~\ref{conf_leak_equality}, $\mathcal{A}$ can further assign a unique identifier to any (abstract) chunk representation so that the identifiers of two chunk representations are equal iff their represented contents are identical. Without loss of generality, we assume that $\mathcal{A}$ assigns identifiers of the form $\hat{k} = (\hat{k'}, \hat{v})$, where $\hat{k'}$ is a value of length $R$ chosen uniformly at random and $\hat{v}$ is a value chosen uniformly at random whose length equals the represented chunk's length in case of a leaf chunk or $y \cdot R$ in case of a superchunk with $y$ children ($\mathcal{A}$ can calculate these values based on $\hat{t}$).

Now we can show that $\mathcal{A}$ can also be provided with the encrypted/auth.~representations of all chunks ever stored without revealing further information about any $m$.

\vspace{-1mm}
\newtheorem*{claimc}{Claim}
\begin{claimc}\label{claimc}
	If $\Pi_E$ is DAE-secure, the probability that an adversary learns anything beyond~\ref{goal_confidentiality} about any content $m$ from the nodes stored in the data structure is negligible.
\end{claimc}
\vspace{-4.5mm}
\begin{proof}
	Assume $\mathcal{A}$ is able to learn something from the encrypted and authenticated chunk representations beyond the aforementioned information with non-negligible probability. First, it is easy to see that the lengths of $\mathcal{A}'s$ previously generated chunk identifiers are equal to the lengths of the actual chunk representations, so she cannot learn anything from the lengths. Being able to learn something from the chunk representations thus implies that she is able to distinguish whether she is given the actual encrypted and authenticated chunk representations or just random strings with the respective lengths.
	
	Let $A$ be her algorithm that on input the information about all contents ever stored in \verb|sec-cs| as stated in \ref{goal_confidentiality} and a complete set of chunk representations of the respective lengths outputs $1$ if the chunk representations are actual chunk representations and $0$ otherwise.
	Now construct an algorithm $B$ with access to an $\textsc{EncMac}$ oracle (with a randomly chosen key) as follows:
	\begin{compactenum}
		\item Initialize a new \verb|sec-cs| data structure and insert all contents $m' \in M$, using the oracle as encryption function, but remembering and reusing oracle outputs instead of asking for same input multiple times.
		\item Pass all information about every content of $M$ as stated in \ref{goal_confidentiality} as well as all (encrypted and authenticated) chunk representations to $A$, yielding output $x$.
		\item Return $x$.
	\end{compactenum}
	If $B$ has access to an actual $\textsc{EncMac}$ oracle, $A$ is given actual chunk representations as created by the data structure. If a random oracle $\$$ with outputs of respective lengths is given to $B$ instead, $A$ gets only random data. If $A$ is able to distinguish both cases with non-negligible probability, $B$ is thus able to distinguish a random oracle from an encryption oracle with non-negligible probability. According to the definitions given in \cite{EC:RogShr06}, $B$ would be an adversary with non-negligible DAE-advantage, which contradicts the assumption that $\Pi_E$ is DAE-secure.
\end{proof}

At this point, $\mathcal{A}$ knows the complete chunk tree $t$ for every content $m$ ever stored in a \verb|sec-cs| instance, including the (encrypted and authenticated) chunk representation of every chunk tree node. We have already seen that $\mathcal{A}$ cannot obtain more information about any content based on this data than stated in Goal~\ref{goal_confidentiality}. Now we show that even metadata (i.e., access patterns from individual operation executions) do not reveal anything more about any individual content. The idea of the proof is as follows: When a data structure operation is executed, $\mathcal{A}$ can only see KVS operation calls made by \verb|sec-cs|. If $\mathcal{A}$ is able to simulate any data structure operation execution to the extent that all KVS operation calls are consistent to a real execution based on information she already has, she does not learn anything from a real operation execution.

Consider a $\textsc{PutContent}(m)$ call. Its execution simply consists of a call of Alg.~\ref{alg:putchunk} with an additional argument $h$.
Knowledge of $h$ allows her to simulate that call, although she cannot provide the content $m$ to Alg.~\ref{alg:putchunk}. The algorithm can be executed consistently given $t$, though: Consider any execution of $\textsc{PutChunk}(m', h')$. If $h' = 0$, the execution corresponds to a leaf chunk of $t$ that is encrypted, authenticated and inserted using $\textsc{Put}$. Since $\mathcal{A}$ already knows the representation of the corresponding chunk, she can simply issue the $\textsc{Put}$ call of line 4. Otherwise, Alg.~\ref{alg:putchunk} performs chunking on the respective chunk, issues recursive calls for the resulting chunks and inserts an encrypted, authenticated superchunk. $\mathcal{A}$ can perform the recursive calls by extracting the children of the current chunk from $t$; she can determine the superchunk representation $c', k'$ from line 10 as it is contained in $t$, and she can issue the $\textsc{Get}$ and $\textsc{Put}$ calls from lines 11--12 since they depend only on $k'$, $c'$ and the children's identifiers.

The call of Alg.~\ref{alg:getchunk} during $\textsc{GetContent}(k)$ is possible for $\mathcal{A}$ due to the same reasons as before. Each individual recursive $\textsc{GetContent}$ call corresponding to a node of $t$ can be trivially simulated by $\mathcal{A}$: The only operation not directly executable by $\mathcal{A}$ is the $\textsc{DecVrfy}$ call in line 3. For the simulation, though, it is sufficient to distinguish three cases. First, $\textsc{DecVrfy}$ fails whenever $k'$ is not a valid authentication tag corresponding to ciphertext $c'$. Since $\mathcal{A}$ knows the correct chunk representation $k'', c''$ for the respective chunk from $t$, she can assume the call to be successful iff $c'' = c'$ except with negligible probability. Second, if $h' > 0$, $k'$ is the identifier of a superchunk, so $v'$ is a list of its children's keys, which she can simply extract from $t$. Only if $h' = 0$, $\mathcal{A}$ fails to compute $v'$. In this case, however, all subsequent operations performed by a benign client are exclusively local (without any feedback to the storage backend), so the simulation is consistent and sound from an adversary's perspective.

Thus, $\mathcal{A}$ is able to perform consistent simulations of all operations, which proves Goal~\ref{goal_confidentiality}.

Note that choice of $\mathcal{C}$ defines a trade-off between confidentiality and storage efficiency. If $\mathcal{C}$, e.g., was WFC, strongest security guarantees could be achieved (although this would fail to achieve storage efficiency): \ref{conf_leak_chunking} would not leak any information at all and \ref{conf_leak_equality} would only leak equality of complete contents. If SC was used, \ref{conf_leak_chunking} would still not leak any information as its output depends only on a content's length which is covered by \ref{conf_leak_length}, but equality of (small) chunks naturally provides an adversary with more information. For CDC schemes, after all, \ref{conf_leak_chunking} becomes relevant as chunk boundaries are computed based on plaintext content parts. Precise security implications depend on the specific scheme and cannot be determined in general. An analysis for one scheme is given in \cite{svnpaper}.

\section{Implementation}\label{implementation}

To ease adoption in practice and to allow for an empirical evaluation (see Sec.~\ref{evaluation}), we have created an implementation. The data structure including our chunking scheme ML-* is wrapped into a flexible Python module named \verb|seccs| available for download in the Python Package Index~\cite{pypi} or via \verb|pip install seccs|.
Unit tests verifying the implementation's correctness w.r.t.~Goals~\ref{goal_put_storage_simple}, \ref{goal_put_storage_efficiency} and \ref{goal_get_authenticity} are bundled with the module.

Since we could not find a sufficiently efficient rolling hash Python implementation, we also developed a rolling-hash-based chunking module \verb|fastchunking| compatible to \verb|seccs| and available for download in PyPI as well. It is a wrapper for parts of the highly efficient \emph{ngramhashing} C++ library~\cite{ngramhashing} by Daniel Lemire and thus able to outperform pure-Python implementations.

\section{Evaluation}\label{evaluation}

We present an extensive evaluation of \verb|sec-cs|'s storage efficiency consisting of two parts: Sec.~\ref{evaluation_seccs} confirms our stated efficiency goals both analytically and empirically and Sec.~\ref{evaluation_comparison} compares the performance of \verb|sec-cs|'s novel chunking scheme ML-* to other approaches.

\subsection{Assumptions}\label{evaluation_assumptions}

To provide concrete numbers, we make some assumptions about the implementation of \verb|sec-cs|. We assume that a deterministic authenticated encryption scheme with length-preserving ciphertexts and $(D=32)$-bytes MACs is used (e.g.~AES-SIV-256), resulting in a constant storage requirement of $R = D = 32$ bytes for chunk references.

By \emph{storage costs}, we refer to the storage consumption of the used KVS (see Sec.~\ref{prerequisites}) for some state. To be independent of any specific backend data structure, we ignore any overheads and roughly estimate storage costs as the sum of the sizes of the KVS's elements, where an element's size is the sum of the sizes of its key and value.

\subsection{Storage Performance of sec-cs}\label{evaluation_seccs}

To complement the proofs from Sec.~\ref{basic_correctness}, we evaluate \ref{goal_put_storage_efficiency} in detail. We do not continue an asymptotic discussion but work out concrete storage costs to judge suitability of \verb|sec-cs| in practice.
The analytical deduction is given below and its validity is confirmed empirically afterwards.

\subsubsection{Analytical Evaluation}\label{storage_analytical}

Let $m'$ be a content consisting of random bytes that is already present in \verb|sec-cs|. Goal~\ref{goal_put_storage_efficiency} states that insertion of $m$ differing only in a sequence of $\delta$ bytes should cause storage costs in $\mathcal{O}\left(\delta + \log |m|\right)$. To work these costs out more precisely, we analyze the border cases first.

\paragraph{Border Case: $\delta = |m|$}

If $\delta = |m|$, chunk trees $t$ and $t'$ of $m$ and $m'$ are not expected to share any nodes, so storage of $m$ should cause costs in $\mathcal{O}\left(|m|\right)$ acc.~to the stated goal. (Note that this case also covers Goal~\ref{goal_put_storage_simple}.) Since leaf chunks have average length $S$, the expected number of leaf nodes of $t$ is $\frac{|m|}{S}$. For the same reasons as in the proof in Sec.~\ref{basic_correctness_content_insertion} ($S > 2R$ implies superchunks are expected to have $\geq 2$ children), $t$ has more expected leaf than superchunk nodes, so its total expected number of nodes is:
\small\begin{eqnarray}
	\textsc{ExpN}^{\mathcal{C}}_{|m|} &\leq& \left\lceil \frac{2 \cdot |m|}{S} \right\rceil
\end{eqnarray}\normalsize

Every chunk tree node has an expected size of $S$ according to Sec.~\ref{recursive_chunking} and is stored under a $D$-byte digest, so storage costs for this case are as follows:
\small\begin{eqnarray}
	\textsc{Storage}^{\mathcal{C}}_{|m|} &=& \left(D + S\right) \cdot \textsc{ExpN}^{\mathcal{C}}_{|m|}\label{eq:storage}
\end{eqnarray}\normalsize

\paragraph{Border Case: $\delta = 1$}

If $\delta = 1$, $m$ and $m'$ differ only in 1 byte. Here, storage costs depend on ML-*'s underlying chunking strategy $\mathcal{C}$.
In case of $\mathcal{C} = \textrm{SC}$, $t$ and $t'$ differ in exactly one node at each level, each having size $\leq S$, which trivially results in the following storage costs:
\small\begin{eqnarray}
	\textsc{AddStrg}^{\textrm{SC}}_{h} \leq \left(D + S\right) \left(h + 1\right)
\end{eqnarray}\normalsize

For $\mathcal{C} = \textrm{CDC}$, however, the situation is more complex. First, $\delta = 1$ allows for shifting in case of CDC. Second, $t$ and $t'$ might differ in more than one node at each level.
The reason is that the modification of a single byte might change up to $W$ chunk boundaries at each level of the chunk tree, probably causing extra chunks to be inserted as well. To determine the total average number of chunk tree nodes that are inserted in this case, we analyze how many new chunks are created at each chunking level.

Let us fix some height $h', 0 \leq h' \leq h$. At height $h'=h$ we only have a single (root) chunk that represents the whole content $m$. At height $h' < h$ we deal with chunks of average length $\frac{S^{h'+1}}{R^{h'}}$ (a position is a boundary with probability $\frac{R^{h'}}{S^{h'+1}}$) according to Sec.~\ref{recursive_chunking}. When considering a fixed $W$-byte window containing the changed byte, the probability that this window yields a chunk boundary in $m$ but did not yield a chunk boundary in $m'$ is $\left(1-\frac{R^{h'}}{S^{{h'}+1}}\right) \cdot \left(\frac{R^{h'}}{S^{{h'}+1}}\right)$. The same probability holds for the case in which a chunk boundary present in $m'$ is not present in $m$ anymore. As the resulting new chunks and the chunk that has to be inserted anyway at this level are not necessarily consecutive\footnote{Changing a single byte might change up to $W$ boundaries in an area of $W$ bytes after the change position, but only the chunk before the first boundary contains the change. If two consecutive boundaries in this area exist in $m'$ and $m$, the chunk in between is unchanged, but might be followed by changed chunks if further boundaries are changed.}, creation (omission) of a chunk boundary in contrast to $m'$ might cause up to 1 (2) additional chunks at height $h'$, respectively.

Since there are up to $W$ window contents containing the changed byte at each chunking level and each position yields 1 or 2 new chunks each with probability $\left(1-\frac{R^{h'}}{S^{{h'}+1}}\right) \cdot \left(\frac{R^{h'}}{S^{{h'}+1}}\right)$, we expect up to $1+(1+2)W\left(1-\frac{R^{h'}}{S^{{h'}+1}}\right) \cdot \left(\frac{R^{h'}}{S^{{h'}+1}}\right)$ new chunks at height $h'$. As we have a single changed chunk at height $h$, we get the following upper bound for the expected number of chunk tree nodes differing between $t$ and $t'$:
\small\begin{eqnarray}
	\textsc{ExpNN}^{\textrm{CDC}}_{h} \leq 1 + \sum_{h'=0}^{h-1} \left(1 + 3 W \left(1-\frac{R^{h'}}{S^{{h'}+1}}\right) \frac{R^{h'}}{S^{{h'}+1}} \right)\hspace{-2mm}\label{eq:expnewn}
\end{eqnarray}\normalsize

While chunking is performed in a way that achieves an average size of $S$ for every chunk when applied to random content, we cannot assume an average size of $S$ for \emph{additional} chunks created for inserting $m$ in addition to $m'$. The rationale is that new chunks are created from existing chunks \emph{not} chosen uniformly at random: A random position in a content is more likely to hit large chunks than smaller ones as more positions are covered by them.

Consider the chunks of $m'$ at some fixed height $h'$. As each byte position is a height-$h'$ chunk boundary with probability $p = \frac{R^{h'}}{S^{{h'}+1}}$, the probability for a chunk having length $c$ is $\left(1-p\right)^{c-1} \cdot p$. Note that $0 < p < 1$ since $S > R$. As we expect an average number of $|m'| / \frac{1}{p}$ chunks at height $h'$ in total, the expected number of length-$c$ chunks at height $h'$ is
$|m'|p \cdot \left(1-p\right)^{c-1} \cdot p$.
Since each of those chunks covers $c$ bytes and since the total content length is $|m'|$, the fraction of the content that is covered by chunks of length $c$ is:
\small\begin{eqnarray}
  \left( \frac{c}{|m'|} \right) \cdot |m'|p \cdot \left(1-p\right)^{c-1} \cdot p &=& cp^2 \cdot \left(1-p\right)^{c-1}\nonumber
\end{eqnarray}\normalsize

Thus, the expected \emph{length} of a height-$h'$ chunk at a position chosen uniformly at random is:
\scriptsize\begin{eqnarray}
  && \sum_{c=1}^{|m'|} c^2 p^2 \cdot \left(1-p\right)^{c-1}
	~~=_{p < 1}~~\frac{p^2}{1-p} \sum_{c=1}^{|m'|} c^2 \cdot \left(1-p\right)^{c}\nonumber\\
  &\leq& \frac{p^2}{1-p} \sum_{c=1}^{\infty} c^2 \cdot \left(1-p\right)^{c}
  ~~=_{|1 - p| < 1}~~\left(\frac{p^2}{1-p}\right) \left( \frac{p^2 - 3p + 2}{p^3} \right)\nonumber\\
  &=& \frac{p^2 - 3p + 2}{p(1-p)}~~=~~\frac{2}{p} - 1~~=~~2 \left( \frac{S^{{h'}+1}}{R^{h'}} \right) - 1\nonumber
\end{eqnarray}\normalsize

As height-$h'$ superchunks store $R$-byte references to height-$(h'-1)$ chunks of avg.~length $\frac{S^{h'}}{R^{h'-1}}$ and as size equals length for leaf chunks, the exp.~\emph{size} of a height-$h'$ chunk at a random position (and thus the exp.~size of any chunk created when inserting $m$) is upper-bounded by:
\small\begin{eqnarray}
	\textsc{ExpChunkSize}^{\textrm{CDC}} &\leq& 2 \cdot S\label{eq:evaluation_rand_chunk_size}
\end{eqnarray}\normalsize

This results in the following storage requirement:
\small\begin{eqnarray}
	\textsc{AddStrg}^{\textrm{CDC}}_{h} \leq \left(D + 2S\right) \textsc{ExpNN}^{\textrm{CDC}}_{h}
\end{eqnarray}\normalsize

\paragraph{Remaining Case: $1 < \delta < |m|$}

In the remaining case, i.e., $m$ differing from $m'$ in a sequence of more than $1$ but less than $|m|$ bytes, both the first and the last byte of the change bytestring affect chunks as in the case $\delta = 1$. For ML-CDC, they are likely to be part of \emph{large} chunks of expected length $2S$ for the same reasons as discussed above. We conservatively estimate that these two bytes cause storage costs of $2 \cdot \textsc{AddStrg}^{\mathcal{C}}_h$ (with $h = \left\lceil \log_{\frac{|m|}{S}} \left( \frac{S}{R} \right) \right\rceil$). The remaining bytes are either part of those chunks (so their storage costs have already been accounted for), or they result in the same chunks that would be created if they were inserted into \verb|sec-cs| as a separate content, causing storage costs up to $\textsc{Storage}^{\mathcal{C}}_{\delta-2}$. Thus, we estimate the total storage requirement as follows:
\small\begin{eqnarray}
	\textsc{DeltaStrg}^{\mathcal{C}}_{h, \delta} &\leq& 2 \cdot \textsc{AddStrg}^{\mathcal{C}}_{h} + \textsc{Storage}^{\mathcal{C}}_{\delta-2}
\end{eqnarray}\normalsize

As $h$ is logarithmic in $|m|$, all cases fulfill Goal~\ref{goal_put_storage_efficiency}.

\subsubsection{Empirical Evaluation}\label{storage_empirical}

\begin{figure*}[t]
	\centering
		\includegraphics[width=0.7\textwidth]{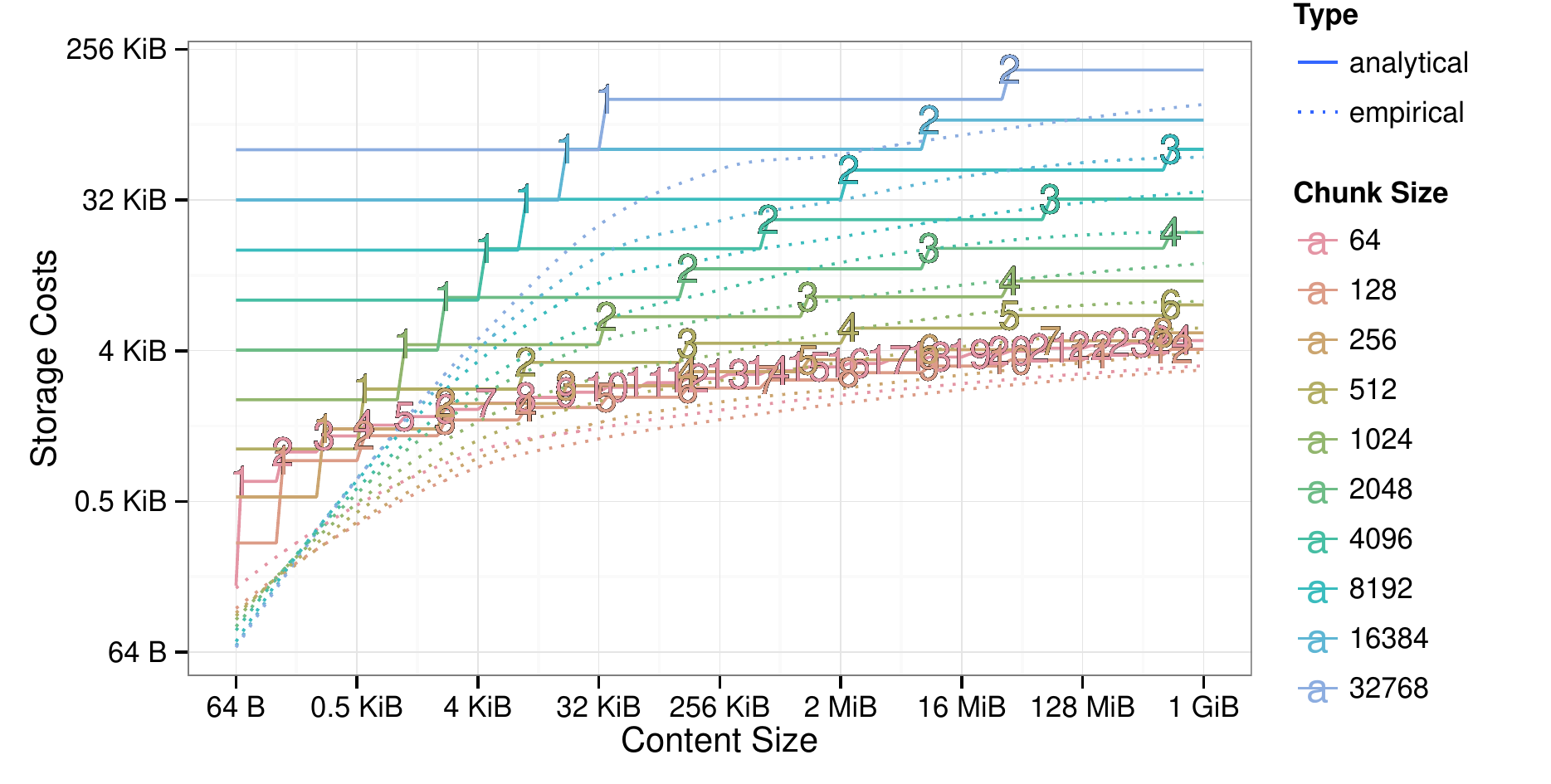}\vspace{-5mm}
	\caption{Storage costs for modified contents ($\delta = 1$)}\label{fig:g2_delta_1}\vspace{-2mm}
\end{figure*}

\begin{figure*}[t]
	\centering
		\includegraphics[width=0.7\textwidth]{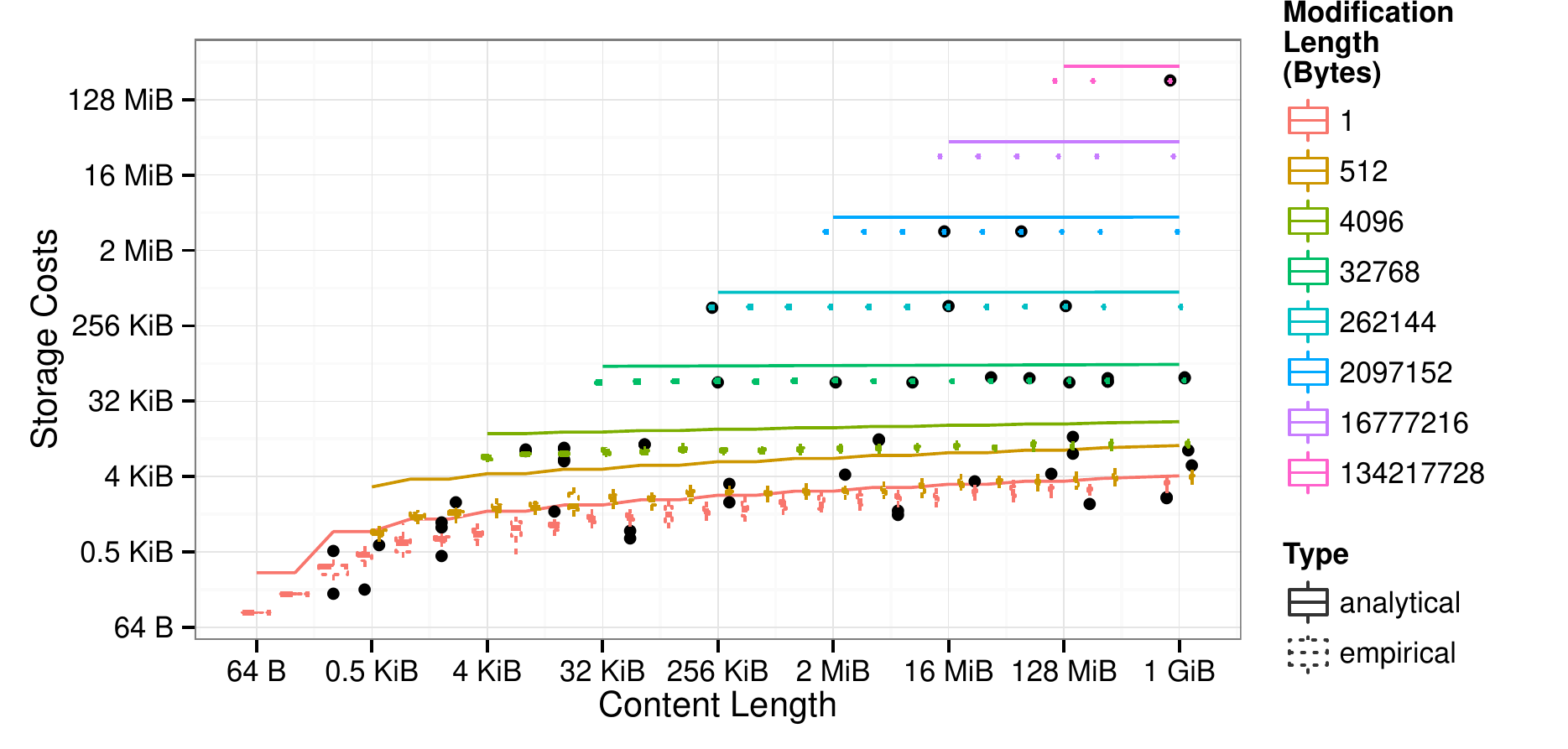}\vspace{-5mm}
	\caption{Storage costs for different modifications ($S = 128$)}\label{fig:g2_delta_all}\vspace{-2mm}
\end{figure*}

As there are no conceptual differences between ML-SC / ML-CDC w.r.t.~to the goal evaluated in this section, we focus on the more interesting ML-CDC scheme in the empirical evaluation. We perform experiments with our implementation (see Sec.~\ref{implementation}) meeting the parameters from Sec.~\ref{evaluation_assumptions}, using a Rabin-Karp-based CDC scheme with window size $W = 48$ bytes (based on evaluation results of the CDC authors~\cite{lbnfs}) and no min/max~chunk size set.

We simulate the scenario from the analytical evaluation: We choose a content $m$ of size $|m|$ uniformly at random, insert it into an empty \verb|sec-cs| instance and remember its state. We replace a randomly chosen $\delta$-bytes substring ($1 \leq \delta \leq |m|$) of $m$ by a different $\delta$-bytes substring chosen uniformly at random, insert the resulting content $m'$ and compare \verb|sec-cs|'s size to the remembered one to measure increase in storage costs. We executed the experiment $20$ times for each combination of content size $|m|$, chunk size $S$ and $\delta$, including border cases $\{1, |m|\}$.

Fig.~\ref{fig:g2_delta_1} compares empirical and analytical results for $\delta = 1$, showing the minimal expected storage overhead for insertion of highly redundant contents.
Solid lines show the calculated relation between content sizes and increase in storage costs for different chunk sizes. While sublinear growth is visible for either chunk size, smaller sizes result in even lower storage costs, unless the chunk size is chosen unreasonably small: For the smallest evaluated chunk size ($64$~bytes), costs incurred by additional superchunk levels outweigh the smaller per-chunk costs.

Threshold content sizes resulting in a respective number of chunking levels are indicated by the positions of the numbers in the chart.
The reason for the leaps at threshold sizes is that our estimate is based on a constant chunk size, while root nodes are smaller in practice (proportional to content size for a fixed tree height), resulting in smaller trees.
Empirical results confirm growth is smoother in practice: Dotted lines show LOESS~\cite{loess} curves fitted to the measured increase in storage costs (smoothing parameter 0.75, degree 2), summarizing its relation to content length for the respective chunk sizes. Results are in line with the calculated upper bounds, confirming sublinear growth in general and least overhead for $S = 128$.

To verify whether the most promising chunk size of $S=128$ bytes also yields small overheads for larger modifications, results for $S=128$ in the general case ($\delta \geq 1$) are shown in detail in Fig.~\ref{fig:g2_delta_all}. Solid lines represent analytical upper bounds for different modification lengths, yielding least costs for $\delta = 1$ (red). Empirical data including outliers (black points) are illustrated as box plots, which are barely visible since they indicate rather small fluctuations in storage costs. The box plots confirm that the analytical bound is conservative, especially for $\delta > 1$: Even outliers are below their corresponding analytical line.

\begin{figure*}[t]
	\centering
		\includegraphics[width=0.625\textwidth]{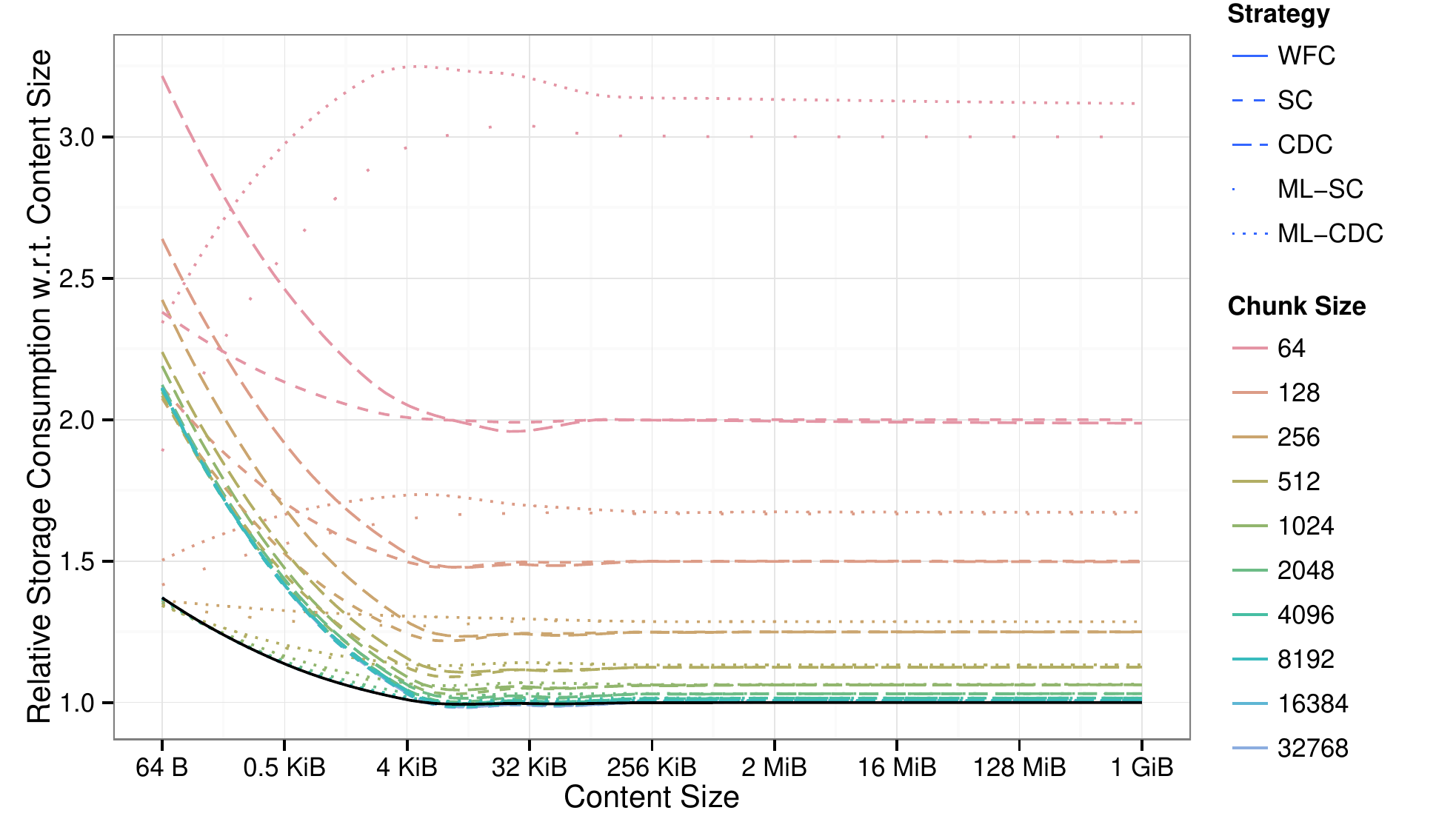}\vspace{-5mm}
	\caption{Storage costs relative to content size}\label{fig:content_relative_storage_consumption}\vspace{-2mm}
\end{figure*}

\subsection{Comparison of WFC / SC / CDC / ML-*}\label{evaluation_comparison}

We compare performance of ML-* empirically to that of the other approaches.
To allow a fair comparison, we evaluate all schemes with \verb|sec-cs|. Note that all schemes are in fact special cases of ML-SC / ML-CDC: Fixing the height of generated chunk trees to 1 results in SC / CDC, respectively, where all metadata representing a content are collected in a single root node; a fixed height of 0 maps each content to a leaf-only tree, corresponding to WFC.

\vspace{-3mm}
\subsubsection{No-Deduplication Storage Overhead}\label{storage_overhead}

Every scheme incurs storage overhead that is necessary to enable data deduplication. In case of WFC, only a small digest (i.e., a hash value) has to be stored for each content for this purpose, resulting in constant overhead. When smaller chunks are stored, overhead is incurred both due to the digests for individual chunks and for storage of a content's representation, i.e., a list of digests of chunks needed to reconstruct it. This overhead is usually compensated by savings from storage of deduplicable contents.

Before comparing the savings achieved by the different schemes, we take a look at the overhead that is incurred when \emph{non-deduplicable} contents are stored. We instantiate \verb|sec-cs| for each chunking strategy and with different chunk sizes and perform the following experiment: We insert a fixed-size content chosen uniformly at random into the empty data structure and measure the \emph{storage expansion factor}, i.e., total storage costs divided by the actual content size. We repeat the experiment for different content sizes, 20 times for each combination of parameters.

The LOESS curves in Fig.~\ref{fig:content_relative_storage_consumption} show the measured expansion factor for different content sizes, which is constant for larger content sizes. It shows that the expansion factor is nearly 1 (in fact, storage overhead is constant: 32 bytes per content) for WFC (solid black line) and slightly above 1 for any scheme with small chunk sizes. With higher chunk sizes, the overhead grows significantly: SC/CDC (dotted lines) have an expansion factor of 1.25 / 1.5 / 2.0 for chunk sizes 64 / 128 / 256; the expansion factor for ML-* (dashed lines) is even higher due to additional storage of non-root superchunk nodes.
Interestingly, ML-CDC produces even more storage overhead than ML-SC. This is due to the concept of ML-* handling each chunk as individual content. In case of ML-SC, the last nodes at each chunk tree height represent chunks smaller than $S$, causing overproportional costs, while any other node is the root of a full, balanced tree. In ML-CDC, these costs are caused recursively by the last nodes of \emph{every} subtree.

\vspace{-3mm}
\subsubsection{Deduplication Overhead without Shifting}\label{dedup_wo_shifting}

\begin{figure*}
	\centering
		\includegraphics[width=1\textwidth]{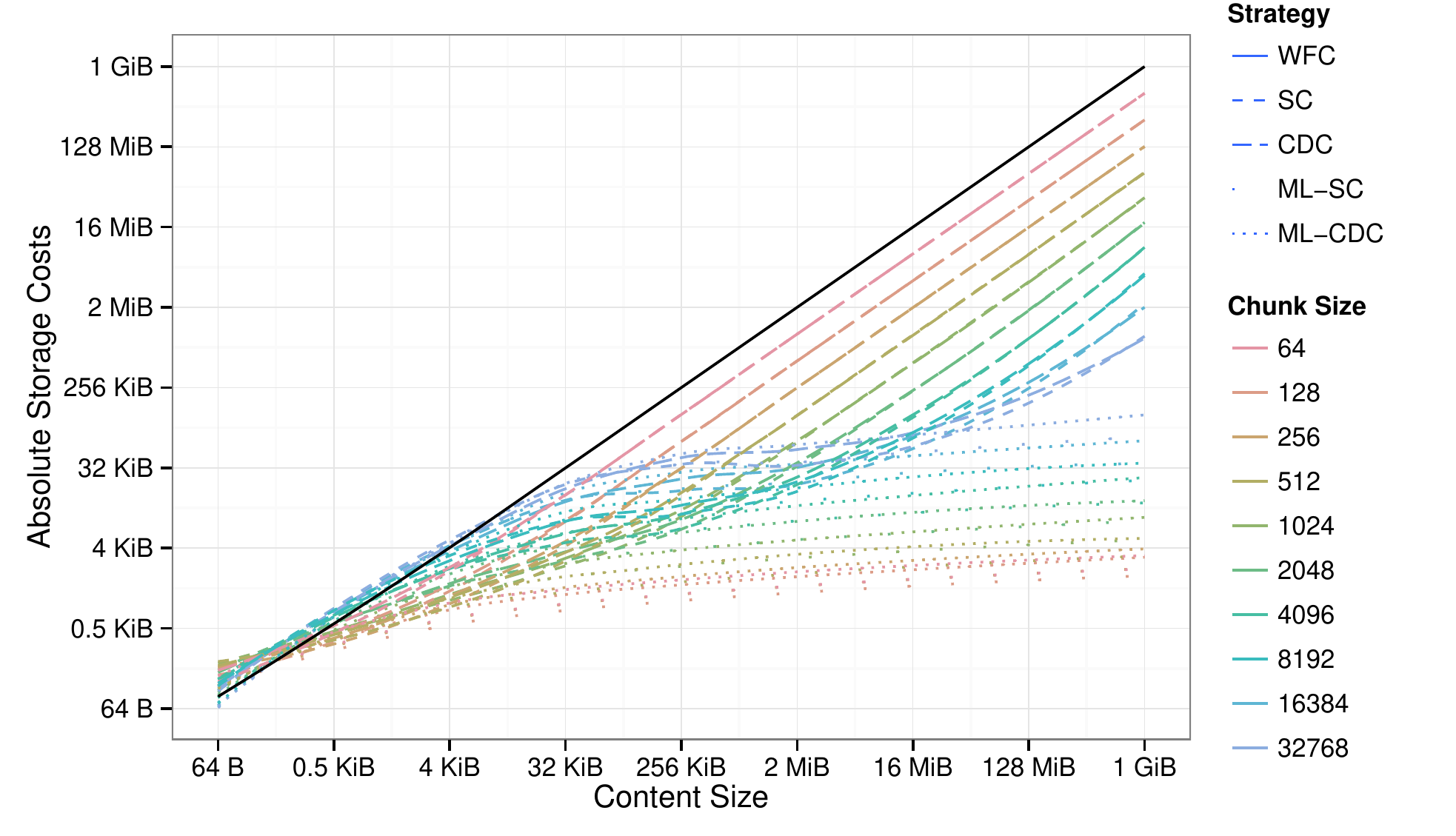}\vspace{-5mm}
	\caption{Storage costs for modified content (overwrite)}\label{fig:storage_costs_overwrite}\vspace{-2mm}
\end{figure*}

\begin{figure*}
	\centering
		\includegraphics[width=1\textwidth]{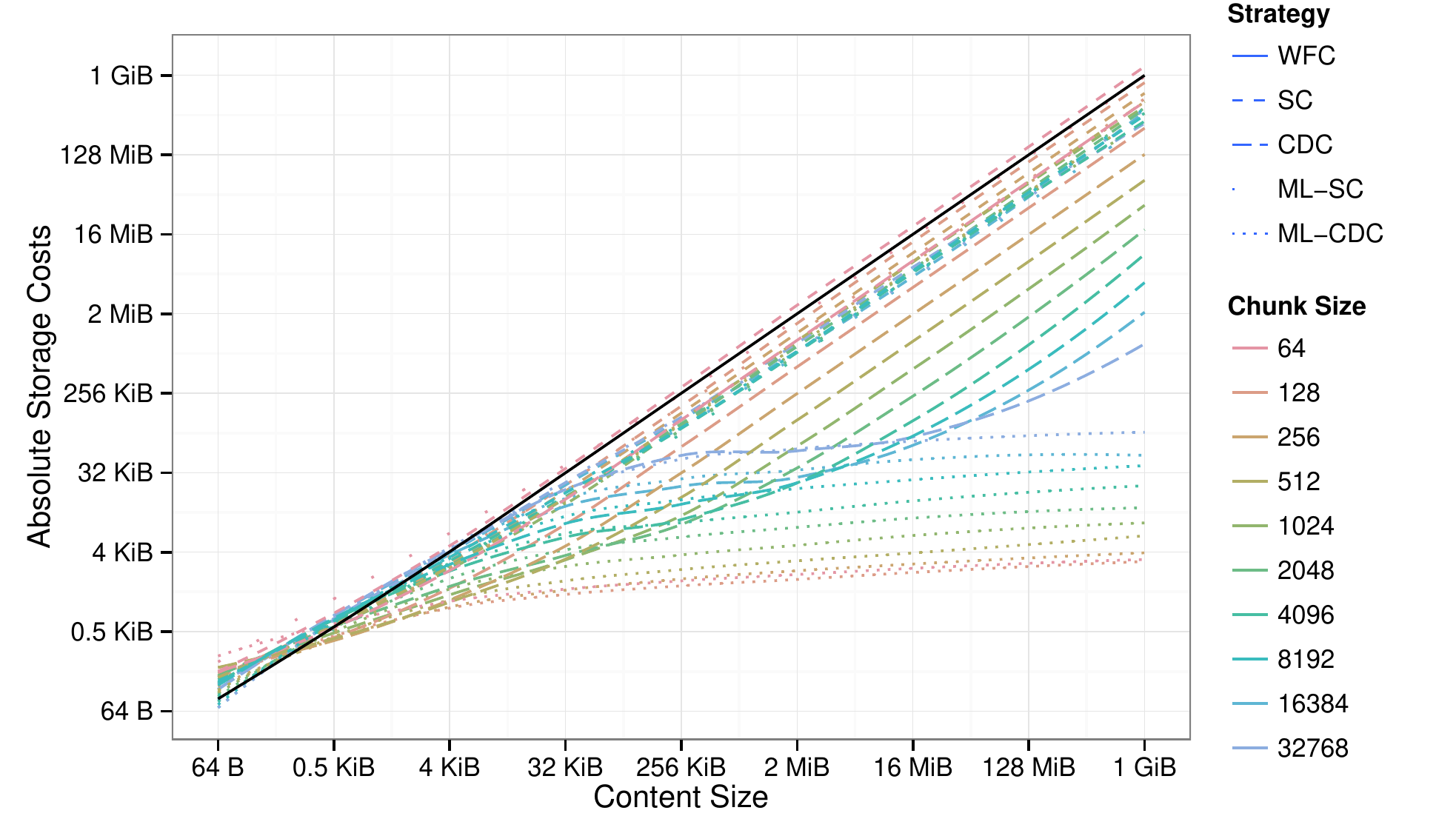}
	\caption{Storage costs for modified content (insert)}\label{fig:storage_costs_insert}
\end{figure*}

To measure possible savings from deduplication, we modify the previous experiment: We insert a second content that differs from the first in only a single byte at an offset chosen uniformly at random. Fig.~\ref{fig:storage_costs_overwrite} shows the total increase in storage costs after the second content has been inserted: For WFC, increase corresponds to the inserted content's size; for SC/CDC, storage costs are only a fraction of that thanks to deduplication, but still linear in the content size. Costs of ML-SC/ML-CDC are orders of magnitude lower and sublinear in the content size.

\vspace{-3mm}
\subsubsection{Deduplication Overhead with Shifting}\label{dedup_w_shifting}

To account for the strengths of CDC, we perform a slight modification of the previous experiment: Instead of overwriting, we \emph{insert} a random byte at a random position, leading to a shift of the remaining content. Results are shown in Fig.~\ref{fig:storage_costs_insert}: As expected, performance of WFC, CDC and ML-CDC is comparable to the previous experiment since they are robust against shifting. SC and ML-SC, however, yield storage costs of about half of the content size for chunk sizes $\geq 256$ (with slight variations in chunk sizes), which corresponds to an expected amount of $50\%$ of the content being \emph{before} the shift position and thus deduplicable. Observe that costs for SC with $S = 64$ are similar to WFC, due to storage expansion factor 2.

\subsubsection{Break-Even Analysis}

\begin{figure*}[t]
	\centering
		\includegraphics[width=0.825\textwidth]{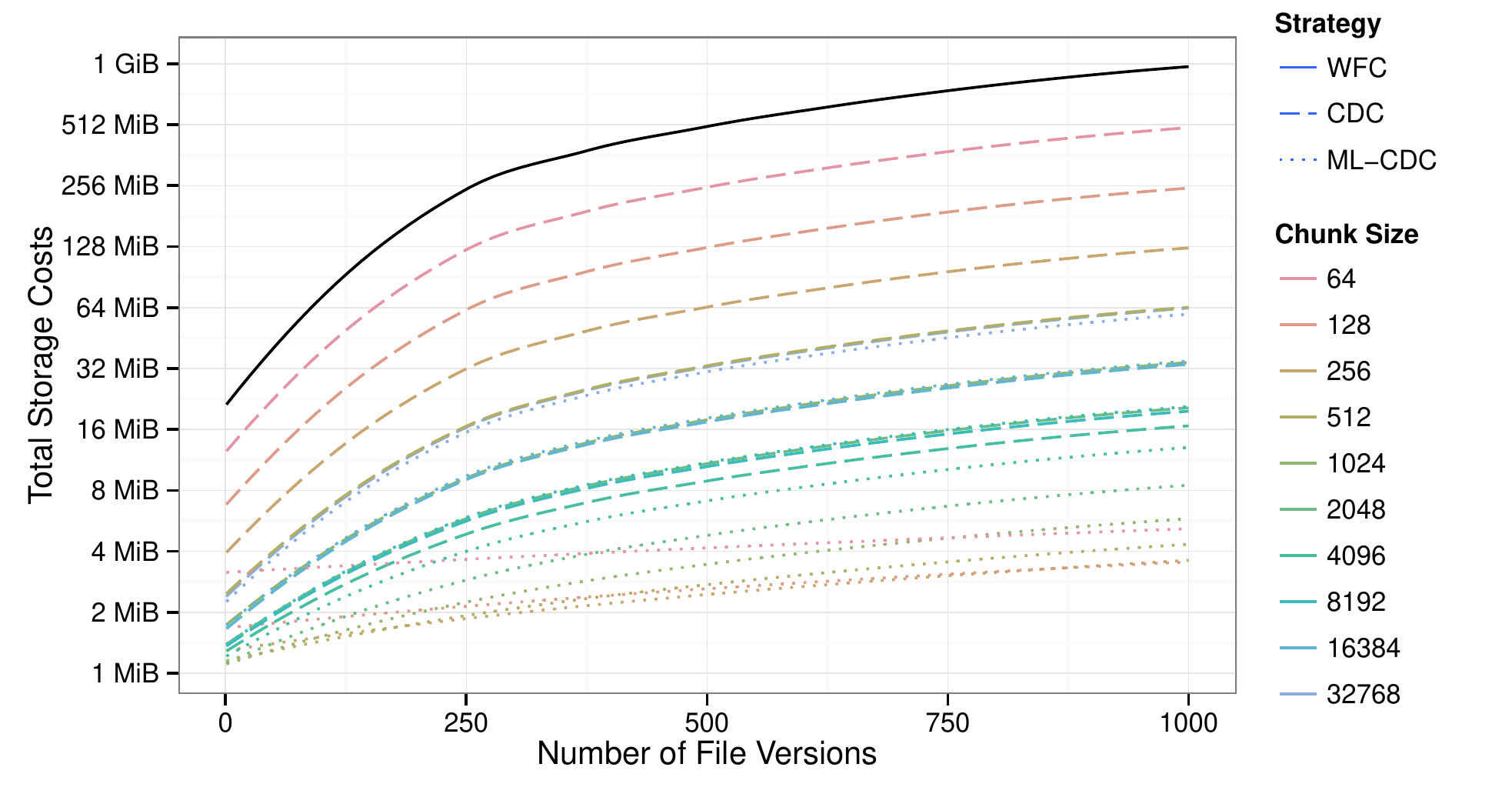}\vspace{-5mm}
	\caption{Storage costs for many similar 1 MiB contents}\label{fig:storage_costs_many_contents}\vspace{-2mm}
\end{figure*}

\begin{figure*}
	\centering
		\includegraphics[width=0.825\textwidth]{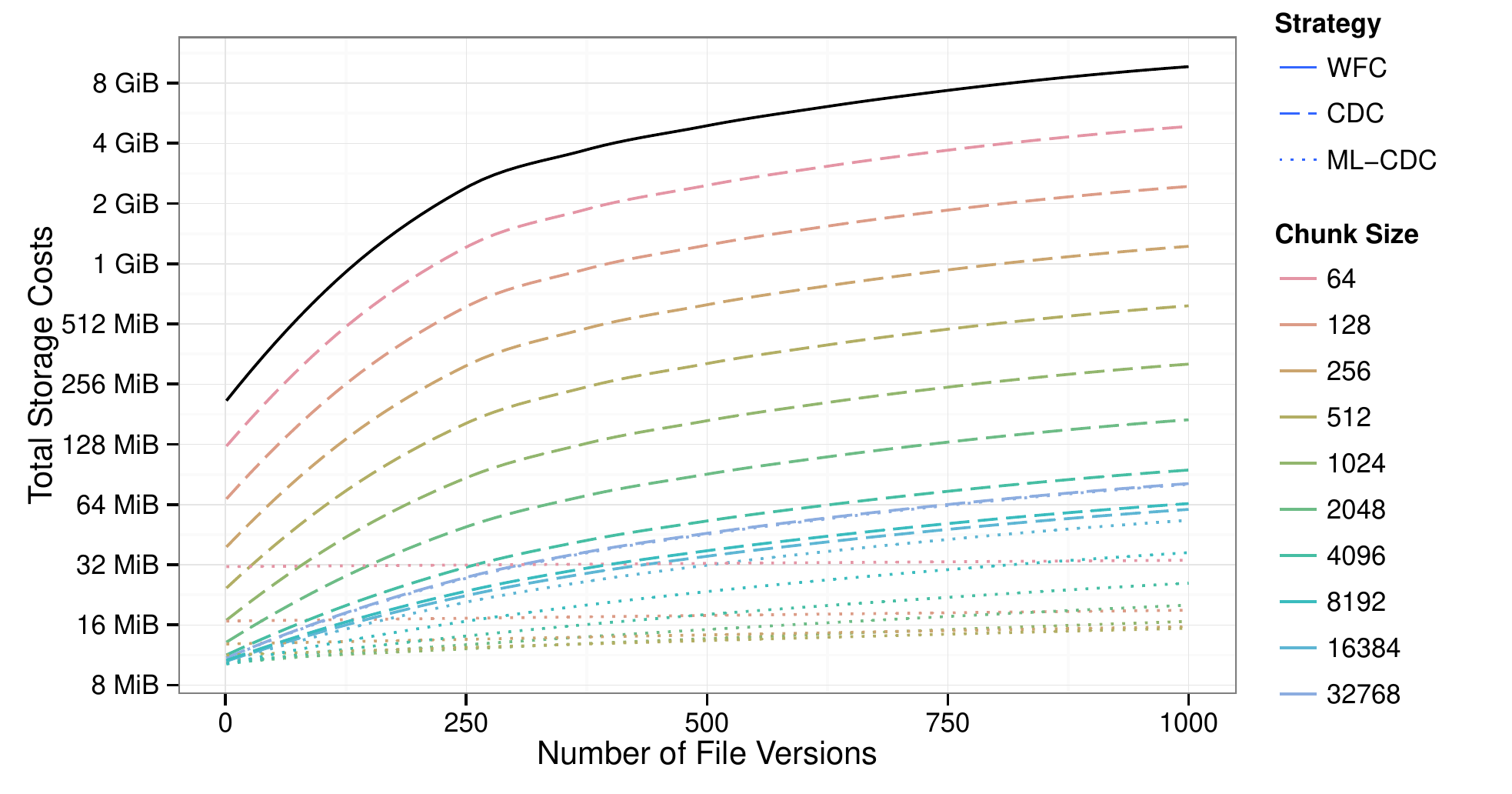}\vspace{-5mm}
	\caption{Storage costs for many similar 10 MiB contents}\label{fig:storage_costs_many_contents10}\vspace{-2mm}
\end{figure*}

Sec.~\ref{dedup_wo_shifting} and~\ref{dedup_w_shifting} have shown that ML-* is significantly more efficient than today's common deduplication strategies when storing contents that differ only slightly from already stored ones, especially for small chunk sizes. However, Sec.~\ref{storage_overhead} has shown that this efficiency comes at the cost of a higher storage expansion factor, i.e., storage of non-deduplicable contents is more expensive in presence of ML-* and small chunk sizes. This raises the question as to whether and when ML-* is preferable. Intuitively, this is the case whenever storage of \emph{many} versions of contents is involved, e.g., in a backup scenario. We investigate this as follows: We start with a fresh \verb|sec-cs| instance containing a single, random 1 (10) MiB content. Then we insert modifications of this content and measure storage costs after each inserted version.

Fig.~\ref{fig:storage_costs_many_contents} (Fig.~\ref{fig:storage_costs_many_contents10}) shows LOESS curves displaying smoothed results over 20 runs for each combination of parameters:
As expected, WFC yields storage costs of about 1 MiB (10 MiB) for every stored content version. Costs for CDC are only a fraction thanks to deduplication: $2048 \leq S \leq 8192$ ($8192 \leq S \leq 16384$) yields lowest costs; for other chunk sizes, CDC incurs significantly higher costs. If only few versions are stored, ML-CDC yields slightly lower costs for \emph{any} chunk size between 256 and 8192 (32768) bytes. The more content versions are stored, the more significant are the savings by ML-CDC: When 125 (250) similar versions are stored, ML-CDC with $256 \leq S \leq 1024$ ($256 \leq S \leq 4096$) requires only half of the storage space as the most-efficient CDC variant; for 1000 versions, costs are orders of magnitude lower.
Note that we omitted results for SC / ML-SC for readability: Due to shifting, they are close to WFC.

\subsubsection{Real-Life Data Comparison}

\begin{figure*}[t]
	\centering
		\includegraphics[width=1\textwidth]{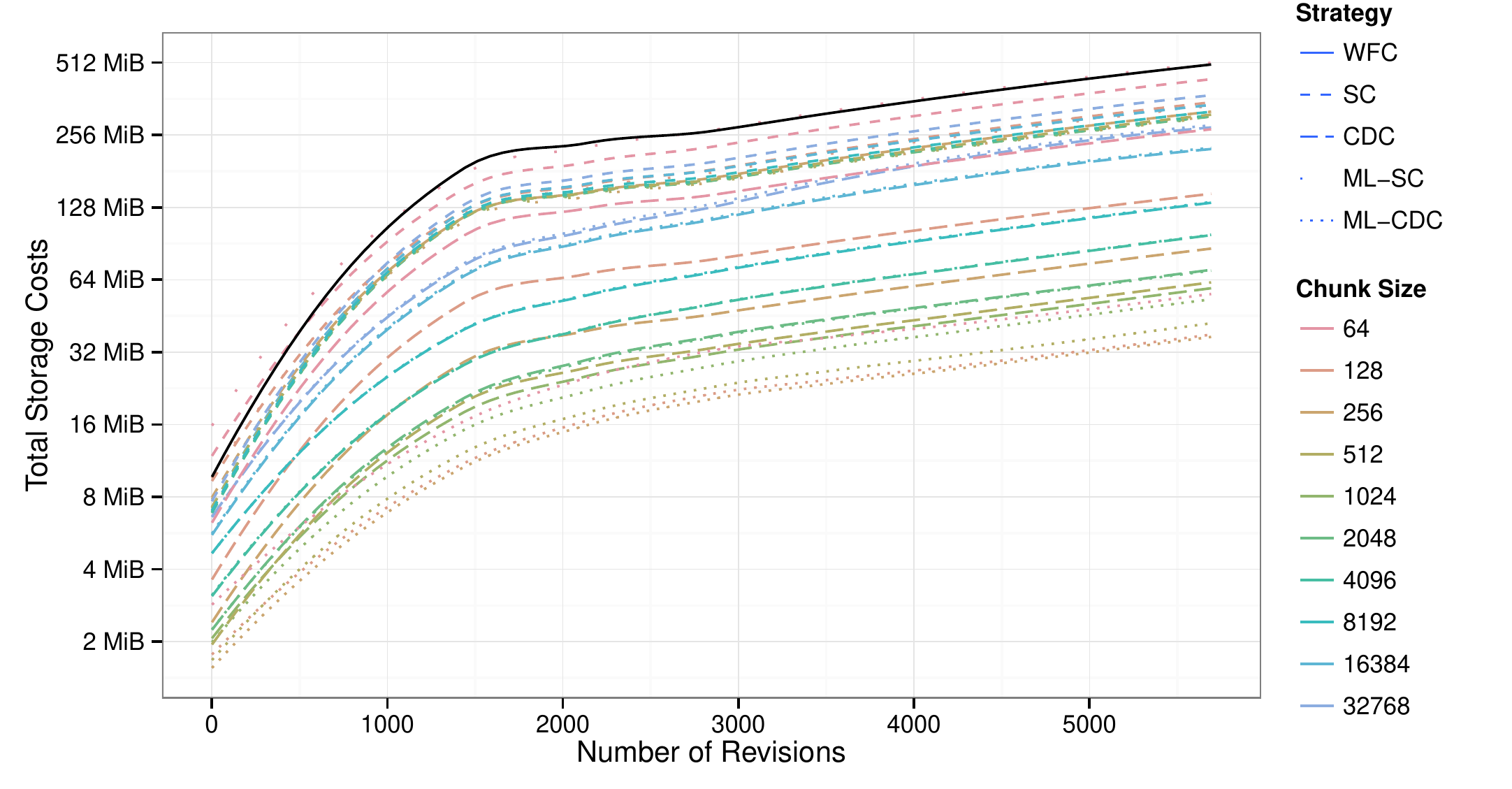}\vspace{-5mm}
	\caption{Storage costs for revisions of Redis git repository}\label{fig:storage_costs_revisions}\vspace{-2mm}
\end{figure*}

While the previous experiments have proven ML-*'s superiority in hypothetical scenarios involving slight changes on random data, it remains to be analyzed how good it performs in real life. We leave an extensive evaluation (e.g., involving a backup system) for future work, but perform the following experiment as a starting point: We insert content versions into \verb|sec-cs| as before, but instead of random data, we insert \emph{all} file contents of all revisions (as of 2016-05-16) of the Redis key-value database git repository~\cite{redis_git} and measure \verb|sec-cs|'s storage costs.

Results (Fig.~\ref{fig:storage_costs_revisions}) are promising: With about $512$ MiB for all $5693$ revisions, WFC causes by far the highest costs. SC / ML-SC can only slightly reduce these costs as they are not robust against shifting. Costs for CDC are lower and range from about $64$ MiB for $512 \leq S \leq 2048$ to $256$ MiB for $S = 64$. For $S \geq 2048$, performance of ML-CDC is comparable to CDC as only a single chunking level is used for most files acc.~to Eq.~\ref{eq:chunking_levels}. For smaller $S$, ML-CDC is significantly more efficient than the other schemes. ML-CDC with $128 \leq S \leq 256$ performed best, causing only about $38$ MiB of total storage costs, which is rather close to the $21$ MiB required by (unencrypted) git.\footnote{Note that comparison between git and \texttt{sec-cs} is unfair: Git accepts additional computational overhead by computing deltas across contents and it applies compression to aggregated contents which is only possible since it does not support encryption, unlike \texttt{sec-cs}.}

\section{Conclusion}\label{conclusion}

We have introduced a data structure for encrypted and authenticated storage of file contents, \verb|sec-cs|, that employs a novel multi-level chunking strategy, ML-*, to achieve storage efficiency. The data structure transparently deduplicates identical parts of file contents without relying on information about relations between them, and achieves storage costs for highly redundant contents logarithmic in their lengths.
We have proven its security and evaluated efficiency extensively w.r.t.~other common deduplication concepts. A ready-to-use, open source Python implementation has been published as part of our work as to allow integration in other software projects.

As next step, we work on a backup system based on \verb|sec-cs| and on an extension that supports partial read and write access to contents, making it suitable as backend for future file systems based on untrusted cloud storage.

\end{document}